\documentclass[10pt]{article}
\usepackage{dcolumn}
\usepackage{bm}
\usepackage{verbatim}       

\usepackage[pdftex]{graphicx}
\usepackage{amsthm}
\usepackage{amssymb}

\usepackage{amstext}
\usepackage{amsmath}
\usepackage{amsfonts}
\usepackage{units}
\usepackage{mathrsfs}
\usepackage{cite}
\usepackage[%
bookmarks=true,
colorlinks,
linkcolor=blue,
urlcolor=blue,
citecolor=blue,
plainpages=false,
pdfpagelabels,
final,
breaklinks=true\between
]{hyperref}
\usepackage{hyperref}
\newcommand{\p}{\partial}

\newcommand{\dd}{{\rm d}}

\newcommand{\bd}{\begin{definition}}                
\newcommand{\ed}{\end{definition}}                  
\newcommand{\bc}{\begin{corollary}}                 
\newcommand{\ec}{\end{corollary}}                   
\newcommand{\bl}{\begin{lemma}}                     
\newcommand{\el}{\end{lemma}}                       
\newcommand{\bp}{\begin{proposition}}            
\newcommand{\ep}{\end{proposition}}                
\newcommand{\bere}{\begin{remark}}                  
\newcommand{\ere}{\end{remark}}                     

\newcommand{\bt}{\begin{theorem}}
\newcommand{\et}{\end{theorem}}

\newcommand{\be}{\begin{equation}}
\newcommand{\ee}{\end{equation}}

\newcommand{\bit}{\begin{itemize}}
\newcommand{\eit}{\end{itemize}}
\newtheorem{theorem}{Theorem}[section]
\newtheorem{corollary}[theorem]{Corollary}
\newtheorem{lemma}[theorem]{Lemma}
\newtheorem{proposition}[theorem]{Proposition}
\theoremstyle{definition}
\newtheorem{definition}[theorem]{Definition}
\theoremstyle{remark}
\newtheorem{remark}[theorem]{Remark}

\begin{document}
\title{Low regularity extensions beyond Cauchy horizons}
\author{M.\ Lesourd\thanks{Black Hole Initiative, Harvard University, Cambridge, MA 02138. E-mail: mlesourd@fas.harvard.edu} \ and E.\ Minguzzi\thanks{Università degli Studi di Firenze, Dipartimento di Matematica e Informatica "U. Dini", Via S. Marta 3,
I-50139 Firenze, Italia. E-mail: ettore.minguzzi@unifi.it}}

\maketitle
\begin{abstract}

We prove that if in a $C^0$ spacetime a complete partial Cauchy hypersurface has a non-empty Cauchy horizon, then the horizon is caused by the presence of almost closed causal curves behind it or by the influence of points at infinity. This statement is related to strong cosmic censorship and a conjecture of Wald. In this light, Wald's conjecture can be formulated as a PDE problem about the location of Cauchy horizons inside black hole interiors.

\end{abstract}

\section{Introduction}
Penrose introduced strong cosmic censorship in the seminal paper \cite{penrose79}. In an elegant article appearing in the same volume, Geroch and Horowitz \cite{geroch79} elaborated on the differences between the strong and weak formulations of cosmic censorship, and sketched possible approaches one could take to prove such statements. In the context of Lorentzian geometry, they formulated a conjecture which they deemed to be in the spirit of strong cosmic censorship. This conjecture  was picked up and sharpened by Wald in Chapter 12 of his landmark book \cite{wald84}.
\begin{quote}
\textsc{Wald's Strong Cosmic Censorship}. Let $(S, h, K)$ be an
initial data set for Einstein's equation, with $(S, h)$ a complete Riemannian manifold
and with the Einstein-matter equations in which matter is described by a quasi-linear, diagonal, second order
hyperbolic system, and where  the stress-energy tensor satisfies the
dominant energy condition. Then, if the maximal Cauchy development of this initial
data is extendible, for each $p \in H^+(S)$ in any extension, either strong causality is violated at $p$ or $\overline{J^-(p)\cap S}$  is non-compact.\footnote{In Wald's book $J^-(p)$ is replaced by  $I^-(p)$, but the two formulations are easily shown to be equivalent.}
\end{quote}

Wald's conjecture seeks to pinpoint why Cauchy horizons form. That is, they are associated with one of two things: the formation of almost closed causal curves, or the possibility that some point on the horizon
 be influenced by spatial infinity. With regards to the basic examples, in Taub-NUT it is the former, and in Reissner-Nordstrom it is the latter. As for the Kerr-Newman family (with $a\neq 0$), both occur.    \\ \indent
With time, the weak and strong conjectures have been translated into PDE problems about the Cauchy development of generic initial data.

\begin{quote}
\textsc{Modern Strong Cosmic Censorship}. Maximal Cauchy developments of \textit{generic} initial data are inextendible within some \textit{regularity} class.
\end{quote}

Remarkable recent work on the subject \cite{cardoso18,dafermos17,kehle20,kehle21,luk17,luk19,vandemoortel21} shows that the SCC is highly sensitive to the setting in which one poses the problem, and the specific definitions of \textit{generic} and \textit{regular} at play. The SCC is now more of a theme of plausible statements (as opposed to a fixed universal statement), each more or less natural relative to the particular setting, and each having varying chances of being true. \\ \indent
In spite of the intricacies that the modern picture has unveiled, the simplicity and clarity of Wald's conjecture, viewed as a structure theorem for Cauchy horizons, grants it an enduring significance. Moreover, with the current emphasis on low regularity extensions, a modern version of Wald's conjecture ought to permit for $C^0$ extensions. \\ \indent
Before going further, we note that Wald's conjecture cannot hold for positive cosmological constant. A counterexample is provided by the non-extremal Reissner-Nordstr\"om-de Sitter spacetime. This is a vacuum solution of Einstein's equation with  $\Lambda>0$ which is strongly causal but for which we can find a spacelike partial Cauchy hypersurface such that for some $p\in H^+(S)$,  $\overline{J^-(p)\cap S}$ is compact (cf.\ Fig.\ \ref{bix}).
\begin{figure}[ht]
\centering
\includegraphics[width=10cm]{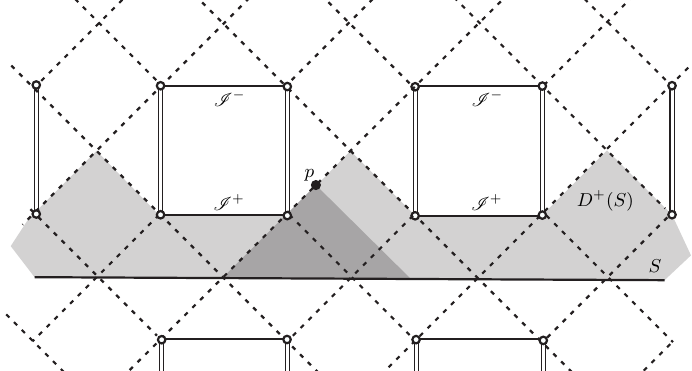}
\caption{The Carter-Penrose diagram of the maximal analytic extension of the non-extremal Reissner-Nordstr\"om de Sitter spacetime. The partial Cauchy hypersurface $S$ admits a point $p\in H^+(S)$ such that $\overline{J^-(p)\cap J^+(S)}$ is non-compact while $\overline{J^-(p)\cap S}$ is compact. The double line is a singularity. } \label{bix}
\end{figure}

It is also worth clarifying the status of Wald's conjecture in the smooth setting. In 2012-3 Etesi posted three preprints which eventually led to a publication \cite{etesi13}, in which Etesi claims to prove a weaker version of Wald's conjecture, where $\overline{J^-(p)\cap S}$ is replaced by $J^-(p)\cap S$. The stronger version is clearly Wald's, since it is easier to prove that a set is non-compact if it can be non-closed.

In a preprint version of our paper (arXiv:2110.07388) we pointed out that these results are not correct, for \cite[Lemma 2.3]{etesi13} contains an error which affects the main claim \cite[Theor.\ 2.1]{etesi13}. The fastest way to realize this is through the following example. Let $(M,g)$ be the 1+1 Minkowski spacetime of coordinates $(t,x)$, $g=-\dd t^2+\dd x^2$ with the timelike lines $x=-1$ and $x=+1$ identified and the point $r=(2,0)$ removed. For  $S=t^{-1}(0)$ and $p=(2.5,0.5)$ we have $p\in H^+(S)$, $J^-(p)\cap S$ is compact but neither $J^-(p)\cap J^+(S)$ nor its closure are compact.

It must be said that Etesi seems to assume that $D(S)$ is the maximal Cauchy development of some initial data, not any Cauchy development (notice that the above example is Cauchy-holed, see  \cite{minguzzi12c} for a definition of hole-freeness). Still, the example signals that there is a problem with the proof of \cite[Lemma 2.3]{etesi13} because that proof does not use  the fact that $D(S)$ is the maximal Cauchy development.\footnote{Indeed, the error can be more precisely identified as follows.
Etesi considers a sequence of curves $\{\lambda_j\}:[0,a_j]\to J^-(x)\cap J^+(S)$ with past endpoints $\lambda_j(0)=p_j$ in $\overline{J^-(x)\cap S}$ and future endpoint $x$. Assuming strong causality of $J^-(x)\cap J^+(S)$, the goal is to infer the compactness of $J^-(x)\cap J^+(S)$ from that of $\overline{J^-(x)\cap S}$. Compactness of $\overline{J^-(x)\cap S}$ is used to obtain an accumulation point $p\in \overline{J^-(x)\cap S}$ for the sequence $\{p_j\}$. At this point, Etesi uses \cite[Proposition 3.31]{beem96} to infer the existence of a limit curve $\{\lambda_j\}\to \lambda$ passing through $p$. So far this is correct. Etesi then cites \cite[Proposition 3.31]{beem96} to claim that the limit curve $\lambda$ also passes through $x$, but this is not implied by the construction in \cite[Proposition 3.31]{beem96}, which only guarantees that $\lambda$ passes through at least one accumulation point, $p$ in this case.}

\begin{remark}
For completeness, we notice that after our preprint was posted (2110.07388v1) Etesi posted a new version (1205.4550v4) in which $J^-(p)\cap S$ is replaced by $\overline{J^-(p)\cap S}$ in both  statements  \cite[Lemma 2.3]{etesi13} and \cite[Theor.\ 2.1]{etesi13}. In the new Theorem 2.1 Etesi is claiming to prove exactly Wald's version. This revision still contains the error pointed out in the previous footnote. For a quicker way to see that the new Theorem 2.1 is incorrect, note that since it is identical to Wald's version as stated above, but does not involve some specification excluding $\Lambda>0$, the statement claimed is manifestly false by the Reissner-Norstr\"om-de Sitter example above (as there $D(S)$ is the maximal Cauchy development).
\end{remark}
Let us now discuss recent developments of non-regular spacetime geometry.

In recent years there have been numerous studies  under a $C^0$ assumption. Fathi and Siconolfi studied the existence of time functions under such regularity\cite{fathi12}, and Chru{\'s}ciel and Grant \cite{chrusciel12} presented the first   investigation of causality theory under the $C^0$ assumption.
 Sbierski \cite{sbierski15,sbierski18} studied causality further, and proved the $C^0$ inextendibility of {S}chwarzschild spacetime. S{\"a}mann \cite{samann16} studied  stable causality and global hyperbolicity and their characterizations.
Galloway and Ling \cite{galloway17b,ling20} proved the $C^0$ inextendibility of AdS spacetime and showed extendibility through the Big Bang of  Milne-like hyperbolic FLRW spacetimes.
Chru\'sciel and Klinger obtained a $C^0$-inextendibility criterion for expanding singularities \cite{chrusciel18}. Galloway, Ling  and Sbierski \cite{galloway18b} proved that in globally hyperbolic spacetimes timelike completeness guarantees $C^0$ inextendibility, a result improved by Minguzzi and Suhr \cite{minguzzi19b}  who showed that global hyperbolicity could in fact be dropped.  Bernard and Suhr \cite{bernard18,bernard19} and Minguzzi \cite{minguzzi17,minguzzi19d} obtained several results on regularity of time functions and Cauchy hypersurfaces under even weaker conditions.
Grant, Kunzinger, S{\"a}mann  and Steinbauer clarified  properties and pathologies of the chronological future under low regularity \cite{grant20}.


In our study, we shall make repeated use of the causality theory developed by the second author in \cite{minguzzi17}.
The regularity permitted in \cite{minguzzi17} is based on the framework of closed and proper cone structures, and the results there are valid for Lorentz-Finsler theory under upper semi-continuity conditions on the cone distribution and Finsler fundamental function. Spacetimes endowed with $C^0$ metrics are special cases of proper cone structures, and so the results on Cauchy horizons developed in that work apply in our case. In fact, it is likely that the results of the present work could hold for the more general proper cone structures of \cite{minguzzi17}, but owing to being motivated by PDEs, we do not see the need to consider lower than $C^0$ regularity.

Although Wald's formulation involves conditions on the initial data (which guarantees the existence of a unique, up to isometry, maximal Cauchy development), one can isolate the purely Lorentzian geometric component of the statement. In particular, what we prove is the following

\begin{theorem} \label{one}
Let $(M,g)$ be a $C^0$ spacetime and let $S$ be a complete Cauchy $C^1$  spacelike hypersurface for $M$. Suppose that $(M,g)$ has a $C^0$ extension $(M',g')$ in which $S$ is still acausal and that $H:=H^+(S,M')\ne \emptyset$.
 Let $p$ be a point in $H$. Then there is a past inextendible geodesic generator $\gamma$ of the horizon with future endpoint $p$ that either
 \begin{itemize}
 \item[(a)] it is past imprisoned in a compact subset of the horizon, and so accumulates in the past direction on a compact region where strong causality (actually distinction)  is violated,
 \item[(b)]  it  escapes every compact set in the past direction and $\overline{J^-(p)\cap J^+(S)}$  is non-compact.
 \end{itemize}
\end{theorem}

Note that although (a) implies that strong causality is violated somewhere on the horizon, this need not occur at $p$. Rather, $p$ is connected to the imprisoning region by way of a null achronal generator.

As for (b), note that strong causality violation can occur without (a) happening at all. For example, every point beyond the Cauchy horizon of maximally extended Kerr is traversed by a closed timelike curve. This means that there cannot be arbitrarily small causally convex neighborhoods around $p\in H^+(S)$. Thus strong causality violation at $p$ can occur in (b) without the imprisonment in (a) occuring.

Note also that we allow for the original spacetime $(M,g)$ to be $C^0$, whereas Wald's formulation presumes that $(M,g)$ is at least $C^2$. Our reason is simply because we can prove the result in such regularity, and that we find it convenient to use the same regularity for the original and extended spacetimes. In this way we need only to prove\footnote{The spacetime previously denoted $M'$ is now denoted $M$, that previously denoted $M$ is now  $\textrm{Int} D(S)$.}

\begin{theorem} \label{two}
Let $(M,g)$ be a $C^0$ spacetime and let $S$ be a complete acausal $C^1$  spacelike hypersurface.  Let $p \in H:=H^+(S)$. Then there is a past inextendible geodesic generator $\gamma$ of the horizon with future endpoint $p$ that either
 \begin{itemize}
 \item[(a)] it is past imprisoned in a compact subset of the horizon, and so accumulates in the past direction on a compact region where strong causality (actually distinction) is violated,
 \item[(b)]  it  escapes every compact set in the past direction and $\overline{J^-(p)\cap J^+(S)}$   is non-compact.
 \end{itemize}
\end{theorem}

The proof of this statement, which largely follows from the theory developed in \cite{minguzzi17}, is given in the next section. This statement is of course consistent with the examples discussed above, and it seems the best version one can obtain in the context of low regularity causality theory. \\ \indent
Of course, some of the terms need to be clarified. For instance, on $S$ the metric $g$ induces a $C^0$ Riemannian metric $h$. By {\em complete} we mean that the closed balls with respect to the associated distance are compact (we do not have at our disposal the Hopf-Rinow theorem in such low regularity). The notion of {\em geodesic generator} will also be clarified in the next section.\\ \indent
Finally, in the spirit of Wald's original conjecture, we will also give necessary and sufficient conditions for $\overline{J^-(p)\cap S}$ to be noncompact. We do this in the context of spacetimes to which we can associate a concept of \textit{null infinity} $\mathcal{J}^+$. This assumption is inspired by Dafermos-Luk's recent groundbreaking study of strong cosmic censorship in vacuum. Their paper deals with strong cosmic censorship as a problem about the stability of Cauchy horizons in a given class of black hole spacetimes, i.e., spacetimes already possessing a future complete $\mathcal{J}^+$ along with an event horizon $\partial I^-(\mathcal{J}^+)\neq \emptyset$. This allows them to sidestep the issue of global existence in the exterior region (i.e., weak cosmic censorship guaranteeing a future complete $\mathcal{J}^+$). In that setting we show the following, where the precise definitions involving $\mathcal{J}^+$ are given below.
\begin{proposition}\label{three}
Let $(M,g)$ be a $C^2$ spacetime admitting a conformal embedding into the unphysical spacetime $(\tilde{M},\tilde{g})$ with null boundary $\mathcal{J}^+\subset \tilde{M}$ such that $\mathcal{J}^+\subset \overline{D^+(S)}$ where the closure is taken in the unphysical spacetime $(\tilde{M},\tilde{g})$. Let $(M',g')$ be a $C^0$ extension of $(M,g)$ with $H^+(S,M')\neq \emptyset$ and $p\in H^+(S,M')$. Then $\overline{J^-(p)\cap S}$ is noncompact if and only if $J^-(p)$ intersects every neighborhood of $\mathcal{J}^+$.
\end{proposition}
Thus, the character of $\overline{J^-(p)\cap S}$ is really about the location of Cauchy horizons (forming dynamically within black holes) relative to $\mathcal{J}^+$. \\ \indent
If, for example, we start with initial data $S$ that differs from a Cauchy surface of the maximal Schwarzschild solution only within a compact subset of the black hole interior, then, even if a Cauchy horizon forms in the development of $S$, $\overline{J^-(p)\cap S}$ remains compact by domain of dependence. \\ \indent
In view of this, Wald's conjecture can be given a modern reformulation as follows.
\begin{quote}
\textsc{Cauchy Horizon Conjecture}. Let $(S,g)$ be a complete Riemannian manifold, and let $(S,g,k)$ be \textit{generic} initial data for the Einstein system, possibly with non-zero cosmological constant $\Lambda$, and possibly with matter sources (in which case we may assume various admissibility conditions on $T_{ab}$). Suppose that the maximal Cauchy development $D(S)=M$ produces an event horizon $\partial I^-(\mathcal{J}^+)\neq \emptyset$. If $M$ admits a $C^0$ extension $M'\supset M$ and a Cauchy horizon $H^+(S,M')\neq \emptyset$, then \[ \overline{J^-(H^+(S,M'))\cap \partial I^-(\mathcal{J}^+)}\: \: \textnormal{is noncompact}\]

\end{quote}
This conjecture can be studied by setting up characteristic initial value problems in which one of the initial data hypersurfaces is an event horizon. This is exactly what is done in \cite{dafermos17} (though in \cite{dafermos17} they study two-ended data meaning that both hypersurfaces are event horizons), and note that the $C^0$ extensions constructed in \cite{dafermos17} do indeed satisfy the above statement.\\ \indent
Finally, we note that since it is in principle possible that portions of $H^+(S)$ differ in the regularity of permitted extensions, the $C^0$ qualifier widens the scope of the conjecture, and that the above statement does not distinguish between $\Lambda<,>,=0$, though indeed its validity may depend on $\Lambda$.

\section{Preliminaries and proof}

Let us introduce some notations and conventions. The spacetime $(M,g)$ is a smooth connected time-oriented  Lorentzian manifold (without boundary) of dimension $n+1$ (working with $C^1$ manifolds would have the same generality as every $C^1$ manifold has a unique smooth compatible structure (Whitney)\cite[Thm.\ 2.9]{hirsch76}). The signature of $g$ is $(-, +, \cdots, +)$. The spacetime is said to be $C^k$ if $g$ is $C^k$, $k\ge 0$.
The symbol of inclusion is reflexive, $X\subset X$. The interior $\textrm{Int} X$ of a set $X\subset M$ might be denoted $\mathring{X}$ for shortness.
 The reader is referred to \cite{minguzzi17,minguzzi19d} for the $C^0$ causality theory and to \cite{minguzzi18b} for all the other conventions adopted without mention in this work.\\ \indent
Two events are chronologically related if there is a piecewise $C^1$ timelike curve connecting them. They are causally related if they coincide or there is a continuous causal curve connecting them. The chronological and causal relations are denoted $I$ and $J$, respectively. We refer to \cite{sbierski15,minguzzi17} for these concepts. It must be recalled that for $C^0$ spacetimes we have  $I\subset  \mathring{J}$, and the equality holds for $g$ locally Lipschitz \cite{chrusciel12,minguzzi17}.

A set is {\em achronal} if no pair of points of it belongs to $I$, and {\em acausal} if no pair of points of it belongs to $J\backslash \Delta$ where $\Delta=\{(x,x): x\in M\}\subset M\times M$ is the diagonal. Lightlike geodesics are defined as in \cite[Def.\ 2.6]{minguzzi17} as  continuous causal curves for which locally no pair of points of it belong to $\mathring{J}$ (this coincides with local achronality only under a locally Lipschitz assumption on $g$).

Sets of the form $J^\pm(p)\backslash \overline{I^\pm(p)}$ have been termed {\em causal bubbles} \cite{chrusciel12}. It is now known that they do not exist iff the Kronheimer and Penrose's causal space condition (push up property) $I\circ J\cup J\circ I\subset I$ holds \cite[Thm.\ 2.8]{minguzzi17} (see also \cite[Thm.\ 2.12]{grant20}), which is the case, for instance, under a locally Lipschitz condition on $g$ \cite{chrusciel12,minguzzi17}.

In the $C^0$ theory the notion of {\em edge} of an achronal hypersurface $S$ does not differ from that of the $C^2$ theory \cite[p.\ 202]{hawking73} \cite[Sec.\ 2.18]{minguzzi18b}. The proof that $\bar S\backslash S\subset \textrm{edge}(S)\subset \bar S$, cf.\ \cite[Prop.\ 2.132]{minguzzi18b} passes to the $C^0$ theory as it uses just the chronological relation and its openness.

An acausal edgeless (and hence closed) set is a {\em partial Cauchy hypersurface}. A {\em Cauchy hypersurface} is a partial Cauchy hypersurface intersected by every inextendible continuous causal curve.

\begin{definition}
Let $(M,g)$ be a $C^0$ spacetime.  The future {\em domain of dependence} or future {\em Cauchy development} $D^+(S)$ of a closed and achronal subset $S\subset M$, consists of those $p\in M$ such that every past inextendible continuous causal curve passing through $p$ intersect $S$. The future {\em Cauchy horizon} is
\[
H^+(S)=\overline{D^+(S)}\backslash I^{-}(D^+(S)).
\]
\end{definition}

We shall need the following result proved in \cite[Thm.\ 2.32]{minguzzi17} for more general ``proper cone structures''. The fact that a $C^0$  spacetime is a  proper cone structure is proved in \cite[Prop.\ 2.4]{minguzzi17} (and similarly under a locally Lipschitz assumption).

\begin{theorem} (Cauchy horizons are generated by lightlike geodesics) \\
\label{juf}
Let $(M,g)$ be a $C^0$ spacetime and let $S$ be a closed and acausal hypersurface (possibly with edge). The set $H^+(S)\backslash S$ is  an achronal locally Lipschitz topological hypersurface and every $p\in H^+(S)$ is the future endpoint of a  lightlike geodesic $\gamma$ contained in $H^+(S)$, either past inextendible or starting from some point in $S$. No two points $p,q\in H^+(S)\backslash S$ can be such that $(p,q) \in \mathring{J}$, thus every continuous causal curve contained in $H^+(S)\backslash S$ is a  lightlike geodesic.
\end{theorem}


Notice that this result does not state  that the  generators cannot intersect other generators in their interior (as we have in the regular theory).

Now, the main strategy is to study the fate of this generator $\gamma$ in the past direction. The first step is to show that it cannot reach the edge of $S$ as $S$ has no edge due to the completeness condition. This fact is non-trivial as $S$ might be ``lightlike'' at the edge and more generally because, though $S$ can be extended as an achronal hypersurface (hence,  locally, a Lipschitz graph), it could be the case that no extension is $C^1$, so that no sensible notion of space metric could make sense over the extension.\footnote{It should be mentioned that in low regularity there are results on the extension of partial Cauchy hypersurfaces to Cauchy hypersurfaces in globally hyperbolic spacetimes \cite[Thm.\ 2.7]{minguzzi19d}. These theorems cannot be used as they require that the hypersurface be extendible at least in a neighborhood of the edge while preserving causality and regularity properties.}

\begin{lemma} \label{lem}
 Let $S$ be a complete acausal $C^1$  spacelike hypersurface on a $C^0$ spacetime. Then $\textrm{edge}(S)=\emptyset$, thus $S$ is closed.
\end{lemma}

\begin{proof}
Let $h$ denote the $C^0$ positive definite metric induced by $g$ on $S$ and let $\rho$ be the distance of  $(S,h)$.

Assume, by contradiction, that there is $p\in \textrm{edge}(S)$. Let $\{x^a\}$ be local coordinates such that $g(p)=-(\dd x^0)^2+\sum_{i=1}^n(\dd x^i)^2$, $\p_0$ is future directed timelike, $x^a(p)=0$ for every $a=0,1,\cdots, n$. Observe that $S$ cannot accumulate on any other point of the axis $x^i=0$ because $S$ is acausal hence achronal.

The bilinear form  $\eta:=b[-a(\dd x^0)^2+\sum_{i=1}^n(\dd x^i)^2]$ evaluated at $p$ is such that for $a<1$  its light cone is narrower than that of $g(p)$, moreover, for $b>1$  the hyperboloid $\{v: \eta(v,v)=1\}$ does not intersect the corresponding unit hyperboloid for $g(p)$. Let, for instance, $b=a^{-1}=2$. These properties extend by continuity to a sufficiently small coordinate cylindrical neighborhood $C$ of $p$. Let $\pi:C\to D$ be the projection on the disk $D=C\cap\{x^0=0\}$, and let the disk have diameter $\delta>0$ in the Euclidan coordinate metric.
The properties imply that if $v\in TC$ is a $g$-spacelike or $g$-lightlike vector then it is a $\eta$-spacelike vector on $C$ and $ \sqrt{g(v,v)}\le \sqrt{\eta(v,v)}\le \vert v^\perp\vert$ where $v^\perp=\pi_*(v)$ is the projection of $v$ on the space orthogonal to $\p_0$ and $\vert \cdot\vert$ is the coordinate Euclidean space norm. Let $q,r\in S\cap C$. If there is curve in $S\cap C$ connecting $q$ to $r$ then this curve has a length smaller than that of its projection to $D$.

Let $A$ be a maximal achronal set containing $S$. Without loss of generality (take a smaller coordinate cylinder $C$ if necessary) we can assume that the set $A\cap C$ is a graph over $D$ (because $A$ is an achronal boundary, so the result follows from \cite[Thm.\ 2.19]{minguzzi17}). Let $q,r\in S\cap C$, and let $\bar q=\pi(q)$, $\bar r=\pi(r)$. The Euclidean-metric-geodesic segment $\bar \sigma:[0,1]\to D$ connecting $\bar q$ to $\bar r$ can be lifted to a curve $\sigma: [0,1] \to A\cap C$ connecting  $q$ to $r$. 

If this curve is not entirely contained in $S$ then there is a maximal segment $\tau$ with starting point $q$ entirely contained in $S$ and converging to some point in $\bar S\backslash S$. Indeed, $\sigma^{-1}(\overline{A\backslash S})$ is a compact subset of $[0,1]$ which admits a minimum value $c>0$, so that $\sigma([0,c))\subset S$ while $\sigma(c)\in \bar S\backslash S \subset \textrm{edge} S$  (clearly $c>0$ because $\sigma(0)=q\in S$ and $S$ does not include its edge. It is also useful to recall \cite[Theor.\ 2.147]{minguzzi18b}). Let $\tau=\sigma\vert_{[0,c)}$ and $\bar \tau=\bar \sigma\vert_{[0,c)}$.

The projection $\bar \tau$ has coordinate-Euclidean length less  than that of $\bar \sigma$ and hence of  $\delta$. However, by the previous bounds, this implies that $\tau$ has $h$-length less than $\delta$ which gives a contradiction with the completeness of $S$, as on $S$ the closed ball of radius $\delta$ centered at $q$ would be non-compact. We conclude that for every $q,r\in S\cap C$ the lift $\sigma$ is actually contained in $S$ which proves that $\rho(q,r)\le \delta$.
However, completeness of $S$ implies that $\rho(q,r)\to \infty$ for $r \to p$, which is not the case as the bound $\delta$ is independent of $r$, that is, it is uniform over $S\cap C$. The contradiction proves that $\textrm{edge}(S)=\emptyset$.
\end{proof}

\begin{remark}
The proof would work replacing ``acausal $C^1$  spacelike hypersurface" with ``achronal hypersurface" provided the distance $\rho(p,q)$ on $S$ were defined as the infimum of lengths $\int_\gamma \sqrt{g(\dot \gamma, \dot \gamma}) \dd t$ over  the Lipschitz curves $\gamma: I\to S$ connecting $p$ to $q$ (the Lipschitz regularity of $S$ does not imply that of curves with image in it, not even under reparametrizations. Notice also that  some pairs of points can be at zero distance, particularly for a null hypersurface, namely $\rho$ would be a pseudometric). The completeness property would have to be understood as before, as compactness of $\rho$-closed balls. Observe that in the above proof the curve $\sigma$ is Lipschitz thanks to the regularity of the projection $\bar \sigma$ and of the achronality of $S$. Clearly, $S$ would be a pseudometric space but would not be a $C^0$ Riemannian manifold since we would not have a $C^0$ Riemannian metric on it, nor a $C^0$ bundle $TS$.  We thank a referee for suggesting the possibility of relaxing the $C^1$ regularity assumption on $S$.
\end{remark}

We are ready to prove Theorem \ref{two} and hence Theorem \ref{one}.

\begin{proof}[Proof of Theorem \ref{two}]
Since  $S$ is a complete acausal $C^1$  spacelike hypersurface, by Lemma \ref{lem}, it is actually closed and $\textrm{edge}(S)=\emptyset$. This means that $S$ is a topological hypersurface and so, by \cite[Thm.\ 2.33]{minguzzi17} the generators of $H^+(S)$ do not reach $S$ in the past direction. Also from \cite[Thm.\ 2.34]{minguzzi17} the set $D^+(S)\backslash S$ is open and $H^+(S)\cap S=\emptyset$. By Theorem \ref{juf} for every $p\in H^+(S)$ there is a lightlike geodesic $\gamma$ with future endpoint $p$ which is past inextendible and contained in $H^+(S)$.
If $ J^-(p)\cap J^+(S)$ has compact closure, then the same is true for the set $J^-(p)\cap H\subset  J^-(p)\cap J^+(S)$,  then $\gamma$ is past imprisoned in a compact set contained in the closed set $H$, and by \cite[Thm.\ 2.22]{minguzzi17} it accumulates on a compact subset of $H$ over which future and past distinction (and hence strong causality) are violated.
\end{proof}
\begin{proof}[Proof of Proposition \ref{three}]
For convenience, write
\begin{equation}
    C\equiv \overline{J^-(p)\cap S}
\end{equation}
By assumption, $S$ has an end $S_{\textnormal{end}}\subset S$, which we define to be that part of $S$ contained in a neighborhood of $\mathcal{J}^+$. Note then that $S$ is a union of a compact set with $S_{end}$ which itself is noncompact.\\ \indent
First consider $\Leftarrow$. If $J^-(p)$ intersects every neighborhood of $\mathcal{J}^+$, then there is a sequence of points $\{p_i\}\in J^-(p)\cap S_{\textnormal{end}}$ which diverges to infinity in $S_{\textnormal{end}}$; that is, $p_i$ eventually escapes every compact set of $S$.

For the converse, we show that if there is a neighborhood of $\mathcal{J}^+$ that does not intersect $J^-(p)$, then $C$ is compact. If such a neighborhood exists, then we can define a closed subset of $S$, $C'$, of the form $C'=S\backslash S_{\textnormal{end}}$ such that $C\subset C'$. Both $C$ and $C'$ are closed, so we need only show that $C'$ is compact. This follows from the fact that $S$ is composed of a compact core and a union of ends.
\end{proof}

\section*{Acknowledgments}
 M.\ Lesourd thanks the John Templeton and Gordon and Betty Moore Foundations for their support of Harvard's Black Hole Initiative.
E.\ Minguzzi thanks GNFM of INDAM for partial support.


\end{document}